\documentclass[10pt,final,twocolumn]{IEEEtran}
\usepackage[nospace,noadjust]{cite}
\usepackage{amssymb,amsmath,amsthm}
\usepackage{graphicx,epstopdf}
\usepackage{float,tikz}
\usepackage{color}
\usepackage{textcomp}
\definecolor{bblue}{rgb}{0,.1,.6}

\usetikzlibrary{arrows,positioning,calc}

\usepackage{amssymb,color}
\usepackage[english]{babel}

\newtheorem{proposition}{Proposition}

\newtheorem{theorem}{Theorem}

\newtheorem{requirement}{Requirement}
\newcommand{\R}{\mathbb R}

\newcommand{\N}{\mathbb N}

\newcommand{\Rplus}{\R_{\ge 0}}

\newcommand{\SP}{\mathcal{SP}}

\newcommand{\cL}{\mathcal L}
\newcommand{\cM}{\mathcal M}

\newcommand{\cZ}{\mathcal Z}

\newcommand{\cK}{\mathcal K}
\newcommand{\cH}{\mathcal H}
\newcommand{\cJ}{\mathcal J}

\newcommand{\jacobian}{{\rm D}}

\newcommand{\x}{\times}

\newcommand{\inv}{^{-1}}

\newcommand{\st}{\,:\,}

\newcommand{\setto}{\rightrightarrows}

\newcommand{\vep}{\varepsilon}

\newcommand{\vhi}{\varphi}

\DeclareMathOperator{\col}{col}

\DeclareMathOperator{\blkdiag}{blkdiag}

\DeclareMathOperator{\argminop}{argmin}

\newcommand{\argmin}[1]{\underset{#1}{\argminop}\,}

\newcommand{\sr}{^\star}

\newcommand{\fid}{\mu}

\newcommand{\id}{z}
\newcommand{\ID}{\cZ}

\newcommand{\vthb}{\vartheta}

\newcommand{\dst}{\delta}

\usepackage{xcolor}
\definecolor{boh}{rgb}{.4,0,.4} 

\newcounter{assc}
\newenvironment{ass}{%  
	\unskip\par\addvspace{.5em}\noindent%  
	\refstepcounter{assc}%
	\textbf{A\theassc)}\itshape% 
}{\par\addvspace{.5em}\ignorespaces} %    

\newcommand{\e}{{\rm e}}
\newcommand{\gravity}{\varrho}

\def\be{\begin{equation}}%
\def\ee{\end{equation}}%
\def\ba{\begin{array}}%
	\def\ea{\end{array}}%

\allowdisplaybreaks

\newcommand{\vsig}{\varsigma}

\newcommand{\dx}{{n_x}}
\newcommand{\du}{{n_u}}
\newcommand{\dy}{{n_y}}
\newcommand{\dw}{{n_w}}
\newcommand{\de}{{n_e}}
\newcommand{\dxo}{{n_0}}
\newcommand{\dchi}{{n_\chi}}
\newcommand{\dchii}{{n_\chi^i}}
\newcommand{\ddchi}{n_\chi}

\newcommand{\dth}{{n_\theta}}

\newcommand{\xx}{{\rm p}}

\begin{document}

	\onecolumn 
	\vspace{4em}
	\begin{quote}
		This is the post peer-review accepted manuscript of: M. Bin and L. Marconi, ````Class-Type'' Identification-Based Internal Models in Multivariable Nonlinear Output Regulation,'' accepted for publication in IEEE Transaction on Automatic Control. The published version is available online at: https://doi.org/10.1109/TAC.2019.2955668
	\end{quote}

\vspace{5em}
\begin{quote}
	\emph{\textcopyright{}~2020 IEEE.  Personal use of this material is permitted.  Permission from IEEE must be obtained for all other uses, in any current or future media, including reprinting/republishing this material for advertising or promotional purposes, creating new collective works, for resale or redistribution to servers or lists, or reuse of any copyrighted component of this work in other works.}
\end{quote}

	\twocolumn

	\title{``Class-type'' Identification-Based Internal Models\\in  Multivariable Nonlinear Output Regulation
	%Oppure: 
	 %Adaptive ``Class-type'' internal models in  multivariable nonlinear output regulation\\
	 %``Chicken-egg dilemma'' in  multivariable nonlinear output regulation and adaptive class-type internal models\\ 
         %%
	%The  ``chicken-egg dilemma'' and identification-based adaptive internal models in multivariable nonlinear output regulation\\
	%Al posto di class-type is potrebbe usare class-shaped
	%}
	}
	\author{Michelangelo~Bin and  Lorenzo~Marconi%
		\thanks{Michelangelo Bin (m.bin@imperial.ac.uk) is with the Department of Electrical and Electronic Engineering, Imperial College London, UK. Lorenzo Marconi (lorenzo.marconi@unibo.it) is with the Department of Electrical, Electronic, and Information Engineering, University of Bologna, Italy.}%
		 
	}%
	
	\maketitle

	\begin{abstract}
	The paper deals with the problem of output regulation in a ``non-equilibrium'' context for a special class of multivariable nonlinear systems stabilizable by high-gain feedback.  A post-processing internal model design suitable for the multivariable nature of the system, which might have more inputs than regulation errors, is proposed. Uncertainties in the system and exosystem are dealt with by assuming that the ideal steady state input belongs to a certain ``class of signals" by which an appropriate model set for the internal model can be derived. The adaptation mechanism for the internal model  is then cast as an identification problem and a least square solution is specifically developed.  In line with recent developments in the field,  the vision that emerges from the paper is that approximate, possibly asymptotic, regulation is the appropriate way of approaching the problem in a multivariable and uncertain context. New insights about the use of identification tools in the design of adaptive internal models are also presented.  
	  \end{abstract}

	\section{Introduction}\label{sec:intro}
We consider  nonlinear systems of the form  
\begin{equation}\label{fmk:sys:plant}
	\begin{aligned}
			\dot x &= f(w,x,u),&  
			y&= h(w,x), & e&=h_e(w,x)
	\end{aligned}
\end{equation}
	with state $x\in\R^\dx$, control input $u\in\R^\du$, measured outputs $y\in\R^\dy$, ``regulation error'' $e\in\R^\de$, and with $w\in\R^{\dw}$ an exogenous signal generated by the ``exosystem"   
\begin{equation}\label{fmk:sys:exo}  
			\dot w = s(w)\,.
\end{equation}
% Let $e= h_e(w,x)$ be a set of $p>0$ {\em error variables} associated to  \eqref{fmk:sys:plant}. 
 The problem of \emph{approximate} output regulation pertains the design of an output feedback regulator of the form
	\[
	\begin{aligned}
	 \dot x_c = f_c(x_c, y)\,,\quad 
	u =k_c(x_c,y)
	 \end{aligned}
	\]
achieving the regulation objective
%\begin{equation*}
%\limsup_{t \to \infty}|e(t)|\le \epsilon 
%\end{equation*}
$\limsup_{t \to \infty}|e(t)|\le \epsilon$, 
%for all possible initial conditions of the system and exosystem in prescribed sets, and
with $\epsilon\ge 0$ possibly a ``small'' number measuring the regulator's asymptotic performance. If $\epsilon=0$, then the regulator is said to achieve \emph{asymptotic regulation}. If $\epsilon$ can be reduced arbitrarily by opportunely tuning the regulator parameters, the regulator is said to achieve \emph{practical regulation}.
% to achieve asymptotic regulation, that is $\lim_{t \to \infty} e(t)=0$ for all possible initial conditions of the system and exosystem in prescribed sets. 
 If the regulation properties are obtained in spite of possible uncertainties in the system  (\ref{fmk:sys:plant}), the problem is referred to as {\em robust output regulation} \cite{Bin2018c}, while the terminology {\em adaptive output regulation} is typically used  in presence of uncertainties in the exosystem (\ref{fmk:sys:exo}).
% The formal definition of robust/adaptive output regulation clearly asks for specifying a topology where the variations of the system/exosystem can range over (see \cite{Bin2018c}).    
 %
An anchor point in the solution of the problem is represented by the steady-state trajectories $(x^\star(t), u^\star(t))$ solution of the so-called {\em regulator equations} 
	\begin{equation} \label{RegEqNL}
	   \dot w =s(w)\,, \quad \dot x^\star =  f(w, x^\star, u^\star)\,,\quad
	  0 = h_e(w,x^\star)\,,
	\end{equation}
with $x^\star$ representing the ideal state trajectory associated with a zero regulation error and $u^\star$ the associated input (often referred to as ``the friend" of $x^\star$). As shown in \cite{Byrnes2003}, indeed, solvability of (\ref{RegEqNL}) is a necessary condition for the problem at hand.  

Regulator structures proposed in the nonlinear context  are typically composed by two units, an {\em internal model unit} and a {\it stabilizing unit}, with a neat, albeit limiting in many contexts, ``role" conferred on the two at the design stage: the former is designed to generate the steady state input  $u^\star(t)$ required to keep the error at zero in steady state,  while the latter is designed to steer the system trajectories to $x^\star(t)$. What makes the design problem particularly challenging is, of course, the fact that $(x^\star, u^\star)$ are unknown as the initial conditions of  (\ref{RegEqNL}) and \eqref{fmk:sys:exo} are such and, in the robust/adaptive case, uncertainties in (\ref{fmk:sys:plant}) and/or (\ref{fmk:sys:exo}) strongly affect the solution of  (\ref{RegEqNL}).  The majority of the current works on the subject have some limiting aspects that is worth pointing out to better frame the contribution of this paper. 

{\em Non-equilibrium context.} Current frameworks typically assume that the solutions of (\ref{RegEqNL}) depend on time through $w(t)$, namely $(x^\star, u^\star) = (\pi(w), c(w))$ for some $\pi$ and $c$. Moreover, further restrictions are usually imposed limiting the class of friends that can be dealt with, as for instance the so-called ``immersion assumption" (the latter even more weakened over the years, see \cite{Huang1994}, \cite{Byrnes2004}, \cite{DelliPriscoli2006}, \cite{Marconi2007}). This assumption, far to be necessary, leads to design principles of the internal model unit just driven by the exosystem dynamics and some appropriate ``distortions"  that, however,  do not completely capture the full nonlinear context. A formal framework to overcome this limitation was given in \cite{Byrnes2003}, where a ``non-equilibrium theory'' for nonlinear output regulation was laid, by asserting that the internal model is in general required to incorporate a mixture of the residual plant's and exosystem's dynamics, in this way making meaningless the distinction between the plant and the exosystem from a  design viewpoint (and thus between robust and adaptive output regulation). 
%In this paper we aim to present a design technique based on the nonequilibrium context, in which the effect of the system and exosystem dynamics on the steady state are jointly considered in the design of the internal model.    

{\em ``Friend-centric" internal models.} Many of the existing regulators are strongly ``friend-centric", namely   the design of the internal model unit is definitely tailored around the specific $u^\star$ resulting from the regulator equations. This, in turn, leads to fragile designs in which unexpected variations of the system/exosystem easily lead to ineffective regulators with unpredictable asymptotic properties. Uncertainties in the system/exosystem are typically handled by parametrising the  internal model in terms of uncertain parameters and by looking for ``adaptive" mechanisms according to the actual regulation error (see e.g. \cite{Serrani2001,DelliPriscoli2006}). This way of proceeding, however, involves a ``quantitative" information about how the uncertainties reflect on the friend that are hard to assume, unless  substantially limiting the topology describing system/exosystem variations.  These difficulties pushed the authors of \cite{Bin2018c} to conjecture that asymptotic regulation in a general nonlinear and uncertain context is unachievable with finite dimensional regulators and to promote approaches looking for approximate regulators, which possibly become  asymptotic if certain fortunate conditions happen.  In general, how a ``qualitative" information about the friend can be transferred into the design of an internal model that behaves ``well" for a ``wide" range of system/exosystem variations is still an open point in literature. 
%In this paper we aim to make a step in this direction by presenting a design solution  not ``friend-centric"  and  built around a ``qualitative" information about the ideal error-zeroing steady state.   

{\em Pre- versus post-processing schemes.}  A taxonomy recently introduced in the literature  regards the distinction between  {\em pre-processing} and {\em post-processing} internal models \cite{Isidori2012,Bin2017}. In the latter, the internal model unit directly processes the regulation error, while the stabilising unit stabilises the cascade of the system driving the internal model unit. In the former, conversely, the two units are somehow ``swapped", with the internal model directly generating the feedforward input and the stabiliser stabilising the cascade of the internal model unit driving the system.  
The regulator structures proposed so far are definitely biased on pre-processing solutions and, as such, limited to deal with single input-single error systems (i.e. $\du=\de=1$) or some ``square" extensions with $\du=\de$ (see, e.g., \cite{Wang2016}).  As observed in \cite{Isidori2012,Bin2017}, post-processing solutions seem more suited to handling general multivariable contexts with possibly $\du>\de$. The latter, in turn, are also more promising to handle contexts in which, besides the regulation errors, also extra measurements are available that do not necessarily  vanish at the steady state. Not surprisingly, the general regulator structure for linear systems is  post-processing \cite{Davison1976}. The drift towards post-processing solutions for nonlinear systems, however, substantially complicates the design of the nonlinear regulator by raising an intertwining  in the design of the internal model and  stabiliser (referred to as {\em chicken-egg dilemma} in \cite{Bin2018b}) not present in pre-processing approaches. To the best knowledge of the authors, a general post-processing nonlinear framework is still unavailable in literature  with just some attempts done in \cite{DAstolfi2017} and   \cite{DAstolfi2015} for simplified exosystems. 

In this paper we propose a design technique based on the aforementioned non-equilibrium context, in which the effects of the system and exosystem dynamics on the steady state are jointly considered in the design of the internal model. The proposed regulator embeds a ``post-processing'' internal model that applies to multivariable systems not necessarily square, and whose construction is not ``friend-centric"  but rather it is based on a  ``qualitative" information on the ideal error-zeroing steady state. 

      \section{Main Result}
       %%%%%%%%%%%%%%%%%%%%%%%%%% 
       %%%%%%%%%%%%%%%%%%%%%%%%%% 
       \subsection{The class of systems}
       We consider a subclass of systems \eqref{fmk:sys:plant} with state $x=\col(x_0,\chi,\zeta)\in \R^\dx$ satisfying the following equations
\begin{subequations} \label{normal_form}
	\begin{align}
	\dot x_0 &= f_0(w,x) + b(w,x)u\label{nf:plant}\\
	\dot{\chi} &= F \chi + H\zeta \label{nf:partial_chi}\\
	\dot{ \zeta} &= q(w,x) + \Omega(w,x) u\label{nf:partial_zeta}   \\
	e  &=   C\chi \,,
	\qquad
	y=\mbox{col}(\chi, \zeta),  %\label{nf:partial_e} 
	\end{align}
\end{subequations}
in which $x_0 \in \R^{\dxo}$, $y\in\R^\dy$, $e\in\R^\de$, $\zeta\in\R^\de$, $u\in\R^\du$, with $\du\ge\de$,   $\chi=\col(\chi^1,\dots,\chi^\de)$, with $\chi^i\in\R^{\dchii}$, $i=1,\dots,\de$, and {$\ddchi^1+\cdots+\ddchi^\de=:\dchi$}, $C:=\blkdiag(C_1,\dots,C_\de)$, $F:=\blkdiag(F_1,\dots, F_\de)$ and $H:=\blkdiag(H_1,\dots,H_\de)$, with $C_i  := \big(
1 \;\;  0_{1\x (\dchii-1)}\big)$ and
\begin{align*}
F_{i} &:= \begin{pmatrix}
0_{(\dchii-1)\x 1} & I_{\dchii-1}\\
0 & 0_{1\x (\dchii-1)}
\end{pmatrix}, &  H_{i}  &:= \begin{pmatrix}
0_{(\dchii-1)\x 1}\\1
\end{pmatrix} .
\end{align*}
%and $C:=\blkdiag(C_1,\dots, C_p)$, with
%\begin{equation*}
%C_i := \begin{pmatrix}
%1 & 0_{1\x (n_\chi^i-1)}
%\end{pmatrix}\,.
%\end{equation*}
The $\chi$ subsystem, in particular, is  described by $\de$ chains of integrators with $\zeta$ entering at the bottom and the regulation error given by the  first components $\chi^i_1$ of each chain $\chi^i$. Hence, $\chi$ and $\zeta$ are linear combinations of the error and its time derivatives. 
%The control input $u(t)$ takes values in $\R^m$, with $m\ge p$, 
The functions $f_0$, $b$, $q$ and $\Omega$ are sufficiently smooth functions, with $\Omega(w,x)\in\R^{\de\x \du}$ denoting the so-called ``high-frequency matrix''. 
The form (\ref{normal_form}) is  representative of different frameworks addressed in literature. For instance, systems having a well-defined vector relative degree with respect to the input-output pair $(u,e)$ and admitting a {\em canonical normal form} fit in the proposed 
framework. In this case the $x_0$ dynamics in (\ref{normal_form}) does not depend on $u$ and it represents the {\em zero dynamics} of the system relative to the indicated input-output pair. On the other hand  (\ref{normal_form}), with a slightly different structure of $\chi$ and of the matrices  $F$ and $H$, is also representative of systems  that are ``just''  (globally) strongly invertible in the sense of \cite{Hirschorn1979,Singh1981} and  feedback linearisable with respect to the input-output pair $(u,e)$ and, as such, can be transformed in {\em partial normal form}, see \cite{Wang2015}.  In this case the dynamics (\ref{nf:partial_chi})-(\ref{nf:partial_zeta}) are the partial normal form of the system  and the subsystem (\ref{nf:plant}) is indeed the whole plant (i.e. $x=x_0$).

%\LM{In the following we are interested in ``semiglobal" results, 
%in which the state of (\ref{normal_form}) is expected to originate from a known compact set $X_0  \subset \R^{\dx}$ defined a priori.} 
%As customary in the literature, the exogenous variable $w$ is v assumed to range in a compact invariant set $W \subset \R^\dw$.
%In the following we are interested in ``semiglobal" results, 
%in which the state of (\ref{normal_form}) is expected to lie in a (arbitrarily large) compact set $X  \subset \R^{\dx}$ defined a priori. {For lack of space, we limit our analysis to the case in which $x(t)$ belongs to the compact space $X$, which is considered fixed from now on. We remark, however that customary high-gain arguments (see e.g. \cite{IsidoriNCS2}) can be used to show that, given an arbitrary compact set $X_0$ of initial conditions, the parameters of the proposed controller can be chosen to ensure the existence of a compact set $X\supset X_0$ having the above invariance property.}
%As customary in the literature, the exogenous variable $w$ is instead assumed to range in a compact invariant set $W \subset \R^\dw$.

We observe that the measurable outputs $y$ are assumed to be linear combinations of the error and its time derivatives,  namely we look  for a {\em partial state feedback} solution. A pure error feedback regulator only processing $e$ can be obtained by replacing the time derivatives with appropriate estimates via  standard high-gain techniques  (see \cite{Teel1995}) whose details are not  presented here.
% \LM{As customary in the literature, the exogenous variable $w$ is   assumed to range in a known compact invariant set $W \subset \R^\dw$. Moreover, we assume the following.}

In the rest of the paper we assume the following.
\begin{ass}\label{ass:regulationEq}
	There exist $\beta_0 \in {\cal KL}$, $\alpha_0>0$ and, for each solution $w$ of (\ref{fmk:sys:exo}), each input $u$, and each   solution $x$ of (\ref{normal_form}) corresponding to $(w,u)$,  there exist  $x^\star_0:\Rplus\to\R^\dxo$ and $u^\star:\Rplus\to\R^\du$ fulfilling 
	\begin{equation}\label{main:eq:regulator_eq}
	\begin{aligned}
	\dot{x}_0^\star &= f_0(w,x^\star)+b(w,x^\star) u^\star\\
	0 &= q(w,x^\star) + \Omega(w,x^\star) u^\star 
	\end{aligned}
	\end{equation} 
in which $x^\star:=(x_0^\star,0,0)$, and
	\begin{equation*}\label{fmk:eq:detectability_e}
	|x_0(t)-x^\star_0(t)| \le \beta_0\big(|x_0(0)-x_0^\star(0)|,t\big) + \alpha_0 |(\chi,\zeta)|_{[0,t)}
	\end{equation*}
	for all $t\ge 0$.  %and for some  not dependent on $(w(t), x(t), u(t))$.
\end{ass}
\begin{ass}\label{ass:stabilisability}
	There exists a  full-rank matrix $\cL\in\R^{\du \x \de}$ such that the (square) matrix $\Omega(w,x)\cL$ is bounded, it satisfies
	\begin{equation*}
	\cL^\top \Omega (w,x)^\top+\Omega(w,x)\cL\ge  I_{\de} 
	\end{equation*}
	for all $(w,x) \in \R^\dw \times \R^\dx$,  and  the map $(\Omega(\cdot)\cL)\inv q(\cdot)$ is Lipschitz.
%	{The maps $q$ and $\Omega$ are Lipschitz and bounded, and there exists a  full-rank matrix $\cL\in\R^{\du \x \de}$ such that the minimum singular value $\upsilon(w,x)$ of the (square) matrix $\Omega(w,x)\cL$ satisfies $\upsilon(w,x)\ge 1$ for all $(w,x) \in {\R^\dw} \times {\R^\dx}$.}
\end{ass}
Equations (\ref{main:eq:regulator_eq}) are the specialisation of the regulator equations (\ref{RegEqNL}) in this non-equilibrium  context.  The steady state $(x^\star, u^\star)$ might be dependent on the initial conditions of the system coherently with \cite{Byrnes2003}.  
%The existence of such a steady state is thus necessary for the problem at hand. 
Condition   A\ref{fmk:eq:detectability_e} asks for {\em uniform (in $u$) detectability} of the ideal steady state $x\sr$
%with the adjective  ``uniform'' that refers to the fact that condition  is required to hold for a generic locally bounded  $u$ 
(see \cite{Liberzon2002}). In case of systems with canonical normal form in which (\ref{nf:plant}) does not depend on $u$, this assumption boils down to a conventional {\em minimum-phase} requirement, far to be necessary  although
typically assumed  in the pertinent literature. %This assumption is far to be necessary for the problem at hand although it is commonly done. 
A\ref{ass:stabilisability}, instead, is a robust stabilisability requirement and it implies that $\Omega(w,x)$ is everywhere full  rank. As a consequence,   $u^\star$ in (\ref{main:eq:regulator_eq}) is given by
\[
 u^\star=- \Omega (w, x^\star)^\top \big( \Omega(w, x^\star) \Omega(w, x^\star)^\top \big)^{-1} q(w, x^\star)\,.
\] 
As clear from A\ref{ass:regulationEq} and A\ref{ass:stabilisability}, we deal with a simplified case in which  a global result is sought under quite restrictive global Lipschitz and boundedness conditions. 
Nevertheless, we remark that the proposed result can be extended to a {\em semiglobal} setting by asking  $w$ to evolve in a compact space  and A\ref{ass:regulationEq} and A\ref{ass:stabilisability} to hold only locally (namely on each compact subset of $\R^\dx$). 
For reason of space and  since the extension follows by well-known arguments (see e.g. \cite{IsidoriNCS2,Teel1995}) we omit this extension and we focus on the new adaptive framework.
       \subsection{The regulator structure}
     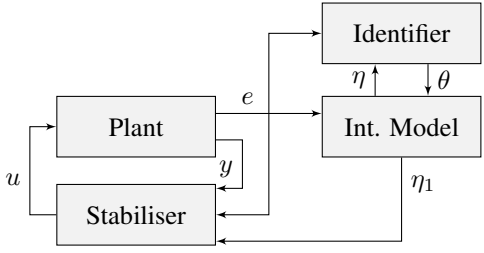
\begin{figure}[t]
     	\centering
     	\tikzstyle{sys} = [draw,inner sep=8,minimum width=6em,fill=black!5!white]
     	\tikzstyle{clk} = [circle,draw,inner sep=2]
     	\tikzstyle{line} = [draw, -latex']
     	\begin{tikzpicture}
     	\node[sys] (plant) at (0,0) {Plant};
     	\node[sys] (intm) [right=4em of plant.east] {Int. Model};
     	\node[sys] (id) [above=1.2em of intm.north] {Identifier};
     	\node[sys] (stab) [below=1em of plant.south] {Stabiliser};
     	\path[line]  ($(plant.east)+(0,.5em)$) -- ($(intm.west)+(0,+.5em)$) node [above,pos=.3] {$e$};
     	\path[line] ($(plant.east)+(2em,0.5em)$)|-(id.west);
     	\path[line] ($(plant.east)+(0,-.5em)$)  -- ($(plant.east)+(1em,-.5em)$)  |- ($(stab.east)+(0,1em)$) node[left,pos=.3] {$y$} ;
     	\path[line] ($(plant.east)+(2em,.5em)$) |- ($(stab.east)+(0,0em)$);
     	\path[line] (intm.south) node [right,yshift=-1em] {$\eta_1$} |-   ($(stab.east)+(0,-1em)$)  ;
     	\path[line] (stab.west) -- ($(stab.west)-(1em,0)$) |- (plant.west) node[midway,left,yshift=-2em] {$u$};
     	\path[line] ($(intm.north)+(-1em,0)$) -- ($(id.south)+(-1em,0)$) node[midway,left] {$\eta$};
     	\path[line] ($(id.south)+(1em,0)$) -- ($(intm.north)+(1em,0)$) node[midway,right] {$\theta$};
     	%\path[line] ($(intm.north)+(1em,1em)$)  -- ($(intm.north east)+(1em,1em)$)  |- ($(stab.east)-(0,.9em)$);	
     	\end{tikzpicture}
     	\caption{Block-diagram of the regulator.}
     	\label{fig:ctrl}\vspace{-1em}
     \end{figure}
       The proposed regulator structure is depicted in Figure \ref{fig:ctrl}. The  post-processing internal model unit has the form
\begin{equation}\label{NLinternalModel}
	\dot \eta = \Phi(\eta,\theta)+ G e, \qquad \eta \in \R^{d\de}
\end{equation}
with $d \in\N$,  $\eta=(\eta_1,\dots,\eta_d)$,  $\eta_i\in\R^\de$, and
\[
\Phi( \eta,\theta) = \left( \ba{ccccc} \eta_2\\  
\cdots \\ \eta_{d}\\ \psi(\eta, \theta)
\ea \right ), \; G = \left( \ba{c} g h_1  I_\de \\g^2 h_2  I_\de \\ \cdots \\g^d h_d I_\de \ea \right ),
\]
 in which $h_i$, $i =1, \ldots,d$, are fixed so that the polynomial $s^d+h_1s^{d-1}+\cdots+h_{d-1}s + h_d$ is Hurwitz, $g>0$ is a  parameter  to be designed, $\psi : \R^{d\de} \times \R^\dth \to \R^{\de}$ is a function   to be fixed, and $\theta\in\R^\dth$,  $\dth\in\N$, is an ``adaptive'' parameter  generated by the {\em identifier} subsystem, whose dynamics is described by
\begin{equation}\label{sys:identifier} 
\begin{aligned}
\dot \id &= \fid(\id,\eta,e), \qquad \theta = \omega(\id), 
\end{aligned}
\end{equation}
in which  $\fid:\ID\x\R^{d\de}\x\R^\de\to\ID$ and $\omega: \ID \to \R^\dth$, with $\ID$ a normed vector space of finite dimension, have to be fixed.  Finally, the (static) stabiliser is taken as  
\be \label{stabiliser}
\ba{rcl}
u &=& \cL \big(   K_\chi  \chi + K_\zeta \zeta + K_\eta \eta_1 + K_w \nu(x^\star,w)    \big),
\ea
\ee
in which the matrices $K_\chi$, $K_\zeta$ and $K_\eta$ are chosen as follows
\begin{align*}
 K_\chi (\ell, \kappa) =\ell K(\kappa),\;\; K_\zeta(\ell) = - \ell I_\de, \;\; 
K_\eta (\ell, \kappa) = \ell K(\kappa) C^\top 
\end{align*} 
with $K(\kappa) = \blkdiag(K^1(\kappa),\dots, K^\de(\kappa))$, where
\begin{equation}\label{d:K}
K^i(\kappa) =-\begin{pmatrix}
c^i_1 \,\kappa^{\dchii} & c_2^{i} \, \kappa^{\dchii-1} & \ldots & c^i_{\dchii} \, \kappa  
\end{pmatrix} 
\end{equation}
for $i=1,\ldots, \de$, in which the coefficients $c^i_j$ are chosen so that the polynomials $s^{\dchii}+c^i_{{\dchii}}s^{\dchii-1}+\cdots+c^i_{2}s + c^i_1$, $i=1,\ldots, \de$, are Hurwitz, and $\ell,\,\kappa>0$ are design parameters to be fixed. The  matrix $K_w$ and the function $\nu$ are introduced for sake of generality and are possibly  zero. These  terms could represent a ``feedforward" contribution added by the designer by employing possible knowledge of $w$ and $x^\star$. Likewise, it could represent a term showing up in the normal form (\ref{normal_form}) after a preliminary feedback of available measurements  that do not vanish in steady state. Similarly to the other matrices in (\ref{stabiliser}),   the gain matrix $K_w$ can depend on $\kappa$ and $\ell$.
%(the example in Section xx  will show a scenario of this kind)
%
%In the previous regulator the design parameters $h_i$, $i=1,\ldots, d$, and  
The  degrees of freedom left to be fixed at this stage are  the dimension $d$ and function $\psi$ of the internal model unit \eqref{NLinternalModel}, the data $(\cZ,\dth,\mu,\omega)$ of the identifier (\ref{sys:identifier}), and the control parameters $g,\, \ell$ and $\kappa$.     
      %%%%%%%%%%%%%%%%%%%%%%%%%% 
       %%%%%%%%%%%%%%%%%%%%%%%%%%    
        \subsection{Design of the internal model as prediction model}\label{sec:designIM}
       A key step in the regulator synthesis is the choice of the internal model \eqref{NLinternalModel} and of its adaptation through the design of the identifier \eqref{sys:identifier}. Consistently with the discussion in Section \ref{sec:intro},  this must be done to achieve a small, possibly zero, asymptotic regulation error in spite of uncertainties involving $(x^\star,u^\star)$ and the underlying dynamics. 
  With an eye to the last equation of (\ref{NLinternalModel}), we can write  
      \be \label{predictionmodel}
       e(t) = \bar c(g)\,  \big (  \dot \eta_d(t) - \psi(\eta(t), \theta(t))  \big )
      \ee
      in which  $\bar c(g) := (h_d g^d)\inv $. Our design strategy to choose $(d,\psi)$ in \eqref{NLinternalModel} and  the identifier \eqref{sys:identifier} pivots around the idea that  $\dot \eta_d(t) - \psi(\eta(t), \theta(t))$ can be interpreted as a ``prediction error'' attained by the ``model'' $\psi$ in relating the ``next derivative'' $\dot\eta_d (t)$ to the ``previous derivatives" $\eta (t)$, and that, by minimising this prediction error, the actual regulation error is also minimised due to (\ref{predictionmodel}). This clearly suggests to look at the problem of choosing $d$ and $\psi$ as an {\em identification problem} and,  by borrowing the notation typically adopted in that  literature \cite{Ljung1999},  to refer to the map  $\psi(\cdot,\theta)$
as the \textit{prediction model} relating the ``input data'' $\eta$ to the ``output'' $\dot \eta_d$, and to the set $\cM:=\{ \psi(\cdot,\theta)\st \theta\in\R^\dth\}$ of all the possible candidate models as the corresponding {\em model set}. The choice of $d$ and of $\psi$ thus must be done in such a way that the attainable prediction error is minimised.  
 Unless relying on  ``universal" infinite-dimensional models, however,  this selection  must be grounded on some  preliminary knowledge about the  \emph{class of signals} to which $\dot \eta_d$ and $\eta$ are expected to belong.  In this context, the steady-state signals $(x^\star, u^\star)$ resulting  from the regulator equations (\ref{main:eq:regulator_eq}) are the anchor point on which that knowledge can be drawn. 
In particular, let  
\begin{equation*}
\eta_1^\star := \Upsilon_{(\ell, \kappa)}(w, x^\star),
\end{equation*}     
in which
%         \[
%         \ba{rcl}
%         \Upsilon_{(\ell, \kappa)}(w, x^\star)&:=& - \left ( \Omega(w,x^\star) \cL K_\eta \right )^{-1} \left ( \, q(w,x^\star) +\right . \\
%          &&\hspace*{2.5cm} \left .  \Omega(w,x^\star) \cL \nu(w,x^\star)\right )
%          \ea
%       \]
   \begin{equation*}
   \begin{aligned}
   &\Upsilon_{(\ell, \kappa)}(w, x^\star) :=\\
   & - \big ( \Omega(w,x^\star) \cL K_\eta \big )^{-1} \big ( \, q(w,x^\star) +   \Omega(w,x^\star) \cL K_\omega \nu(w,x^\star)\big ),
   \end{aligned}
   \end{equation*}
and define recursively $\eta^\star_i$, $i=2,\ldots, d+1$, as\footnote{We denote by $L_g f$ the Lie derivative of $f$ along $g$.}
       \[
         \eta_i^\star := L^{i-1}_{s(w)}\Upsilon_{(\ell, \kappa)}(w, x^\star) + 
         L^{i-1}_{f_0(w,x) +b(w,x)u^\star}\Upsilon_{(\ell, \kappa)}(w, x^\star) .
       \]
       Finally let
       \begin{equation*}
       \dot\eta_d\sr := \eta\sr_{d+1}.
       \end{equation*}
%       \[
%        \dot \eta_d^\star := L^{d}_{s(w)}\Upsilon_{(\ell, \kappa)}(w, x^\star) + 
%         L^{d}_{f_0(\cdot) +b(\cdot)u^\star}\Upsilon_{(\ell, \kappa)}(w, x^\star) .
%       \]
     In view of A\ref{ass:stabilisability} and the definition of  $K_\eta$,  the matrix $\Omega(w,x) \cL K_\eta$ is everywhere invertible and, thus, all the previous quantities are well-defined. Moreover, we observe that the quantities $\eta^\star_i$, $i =1,\ldots, d+1$, depend on the design parameters $\kappa$ and $\ell$ yet to be fixed. The dimension $d$ and the function $\psi$ should be then ideally chosen so that, with $\eta^\star= \col(\eta_1^\star, \ldots, \eta_d^\star)$, the following holds 
     \be\label{idealpsi}
      \dot \eta_d^\star(t) = \psi(\eta^\star(t), \theta^\star(t) ),
     \ee
  for some ``ideal'' $\theta^\star(t)\in\R^\dth$. This, in fact, would make $(x^\star, \eta^\star)$ a trajectory of the closed-loop system in which the associated  regulation error is identically zero. The design of  the pair $(d,\psi)$ so that (\ref{idealpsi}) is fulfilled for all possible steady-state trajectories $(\dot \eta^\star_d, \eta^\star)$, however, is not realistic unless limiting even further the class of treatable nonlinear systems and of manageable  uncertainties on the solution of (\ref{main:eq:regulator_eq}).  Furthermore,  even in the fortunate case in which the ideal relation  \eqref{idealpsi} could be fulfilled with a perfect parametrisation (maybe playing with large values of $d$), this might require an unacceptable complexity of the internal model, and an approximated model with a possibly lower $d$ would be preferable. Along this direction, we rather assume that the designer has a {\em qualitative} knowledge about a ``class" $\cH\sr$ of signals\footnote{Formally, $\cH\sr$ is a subset of the space of functions $\Rplus\to \R^{\de}\x\R^{d\de}$.} to which $(\dot \eta^\star_d, \eta^\star)$ belongs in order to fix a model set $\cal M$ necessarily approximated but optimised for the specific class.  This is the ``modelling part", in which the ``touch" of the designer and the  knowledge on the steady-state trajectories come into play. The class  $\cH\sr$, in turn, is fixed on the basis of the knowledge on the nominal solution $(x\sr,u\sr)$ to \eqref{main:eq:regulator_eq}, and after considering all the expected system/exosystem uncertainties that may affect it. The problem of handling the overall uncertainty on $(x\sr,u\sr)$ is thus transferred to the adaptation side, and the idea of relying on system identification techniques for it is further motivated by the fact that, typically, identification methods structurally manage large classes of signals \cite{Ljung1999}. 
  From now on we suppose that the designer has fixed a class $\cH\sr$  and, accordingly, a model set $\cM$,  so that   the following  assumption holds.
%  with the additional assumption
%  as a Lipschitz function fixed once for all  , by also assuming that its Lipschitz constant is bounded by a fixed number not dependent on the stabiliser parameters  $(\ell, \kappa)$ when they are taken large.    
%  \begin{ass}\label{ass:psiLipschitz}
%	{The map $\psi$ is differentiable and there exist    $l_\psi,\,c_\psi>0$ such that  for all $\ell>1$ and $\kappa>1$, and with $\jacobian\psi$ the Jacobian of $\psi$, the following hold
%	\begin{equation*}
%\begin{aligned}
%	|\psi(\eta', \theta') - \psi (\eta'', \theta'') |  &\leq  l_\psi  \,  |(\eta', \theta') - (\eta'', \theta'') |\\
%	|\jacobian\psi(\eta', \theta')|  &\leq  c_\psi,
%\end{aligned}
%	\end{equation*}
%	for all $\eta',\,\eta'' \in \R^d$  and $\theta',\,\theta'' \in {\R^\dth}$.}
%\end{ass}
  \begin{ass}\label{ass:psiLipschitz}
	The map $\psi$ is Lipschitz and differentiable with a locally Lipschitz derivative, and the Lipschitz constants do not dependent on $\kappa$ and $\ell$. Moreover,  there exists a compact set $H\sr\subset\R^{\de}\x\R^{d\de}$, independent on $\kappa$ and $\ell$,  such that every $(\dot\eta\sr_d,\eta\sr)\in\cH\sr$ satisfies $(\dot\eta\sr_d(t),\eta\sr(t))\in H\sr$ for all $t\in\Rplus$. 
\end{ass}
The previous assumption formalizes the  ``quantitative''  properties required to the members of the class $\cH\sr$ on  which the design of the internal model and the identifier is grounded. In particular, it is asked that the elements of $\cH\sr$ stay in a known compact set $H\sr$, and that the inferred prediction model $\psi$ has some strong regularity properties uniform in the control gains $(\kappa,\ell)$. These requirements, in principle not needed  in the design of the identifier and internal model, are rather needed for the successive embedding of the two units in the overall regulator, as they permit to break the ``chicken-egg dilemma'' and sequence the design of the remaining degrees of freedom.
	
%	\footnote{Namely, the internal model is potentially dependent on the stabilizer, since the choice of $\psi$ is  in principle dependent on $\ell$ and $\kappa$ because $\dot \eta_d^\star$ and $\eta^\star$ depend on those parameters, and the latter depend on the internal model since they are expected to stabilise the cascade of (\ref{normal_form}) with (\ref{NLinternalModel}). Who is designed   first?} that arises in post-processing solutions as delineated in \cite{Bin2018b}. }
	     
	We remark, moreover, that in the ``square'' case, namely when $\du=\de$ in \eqref{normal_form}, the matrix $\Omega(w,x)$ is square and $\kappa$ and $\ell$ do no mix-up with $\Omega(w,x)$, $q(w,x)$ and $\nu(w,x)$ in the  definition of   $\dot\eta_d^\star$ and $\eta^\star$. Therefore  $(\dot\eta\sr_d,\eta\sr)$ can be always bounded uniformly in $\kappa$ and $\ell$ whenever they  are taken larger than $1$ and $K_\omega/(\kappa\ell)$ can be bounded uniformly in $\kappa$ and $\ell$.
%{\color{blue} Eventually remark the case $m=p$, namely the matrix $\Omega(w,x)$ is square (and invertible for all $(w,x) \in X$ due to Assumption A2), and/or $\Omega$ constant.  In those cases $(\kappa, \ell)$ do no mix-up with $\Omega(w,x)$, $q(w,x)$ and $\nu(w,x)$ and the quantitative, and qualitative, impact of  the high-gain parameter is more evident.} 
 
 %%%%%%%%%%%%%%%%
 %%%%%%%%%%%%%%%%
  \subsection{ The design of the identifier}\label{sec:DesignIdentifier}
With $d$ and $\psi$ fixed, we shift our attention to the design of the identifier. The fact that  (\ref{idealpsi}) is not attainable exactly suggests to define a {\it steady state prediction error} as 
\begin{equation}\label{d:varepsilon_star}
\varepsilon^\star(t, \theta) := \dot \eta^\star_{d}(t) - \psi(\eta^\star(t),\theta)
\end{equation}
 and to look for a dynamical system which is able to select the best parameter, say $\theta^\star$, whose corresponding model $\psi(\cdot,\theta^\star(t))$ is, at each $t$, the ``{best}'' model in $\cM$ relating $\dot \eta^\star(t)$ and $\eta^\star(t)$,  minimising in some sense $\varepsilon^\star$.
 As customary in system identification, the meaning of ``best'' in the model selection is based on the definition of a {\em fitness criteria} assigning to each model $\psi(\cdot,\theta)\in\cM$ a suitable and comparable value. In particular, with $C^0(\R^\dth,\Rplus)$ the space of continuous functions $\R^\dth\to\Rplus$, with each pair $(\dot\eta\sr_d,\eta^\star)\in\cH\sr$ we associate the map $\cJ_{(\dot\eta\sr_d,\eta^\star)}:\Rplus\to C^0(\R^\dth,\Rplus)$ given by
\begin{equation}\label{def:cJ}
\cJ_{(\dot\eta\sr_d,\eta^\star)}(t)(\theta) := \int_{0}^t c_\varepsilon\Big(t,s,|\varepsilon^\star(s,\theta)|\Big)ds + c_r(\theta),
\end{equation}
with $c_\varepsilon:\Rplus\x\Rplus\x\Rplus\to\Rplus$ and $c_r:\R^\dth\to\Rplus$ some user-defined positive functions characterising the particular underlying identification problem. More precisely, the integral term of \eqref{def:cJ} measures how well a given choice of $\theta$ fits the historical data, while $c_r(\theta)$ plays the role of a regularisation factor. With $\cJ_{(\dot\eta\sr_d,\eta^\star)}$ we associate the set-valued map $\vthb_{(\dot\eta\sr_d,\eta^\star)}^\circ:\Rplus\setto\R^\dth$ defined as  \begin{equation*}
\vthb_{(\dot\eta\sr_d,\eta^\star)}^\circ(t):=\argmin{\theta\in\R^\dth}  \cJ_{(\dot\eta\sr_d,\eta^\star)}(t) (\theta)\,.
\end{equation*}
%{For simplicity, in the forthcoming developments we suppose that $\vthb_{\eta\sr}^\circ(t)$ belongs to a \emph{known}, yet arbitrary, compact set $\Theta\sr\subset\R^{\dth}$. This condition hides a further qualitative knowledge on the expected class of signal to which $(\eta_d\sr,\eta\sr)$ belongs.}
Once a cost functional of the form \eqref{def:cJ} is defined, the identifier subsystem \eqref{sys:identifier} is constructed to guarantee the existence of an ``optimal'' steady state $z\sr$, which is robustly asymptotically stable for \eqref{sys:identifier}, and whose corresponding output $\theta\sr=\omega(z\sr)$ is a pointwise minimiser of $\cJ_{(\dot\eta\sr_d,\eta^\star)}(t)$, i.e. satisfies $\theta\sr(t)\in\vthb_{(\dot\eta\sr_d,\eta^\star)}^\circ(t)$ for all $t\ge 0$. In particular, the identifier \eqref{sys:identifier} is chosen as a system with state 
\begin{equation*}
z = \col(\xi,\varsigma),\qquad   \xi \in\R^{2\de},\ \vsig\in \cZ_\varsigma,
\end{equation*}
in which $\cZ_\varsigma$ is a finite-dimensional normed vector space, $\cZ=\R^{2\de}\x\cZ_\vsig$,   
and the pair $(\mu,\omega)$ is chosen so that, with $\xi_1,\,\xi_2\in\R^{\de}$ such that $\xi=\col(\xi_1,\xi_2)$, the equations \eqref{sys:identifier} read  as
\begin{equation}\label{s:xi_vsig}
\begin{array}{lcl}
\dot \xi_1 &=& \xi_2- m_1 \, \rho\, (\xi_1 - \eta_d)\\
\dot \xi_2 &=&    \dot \psi(\xi_2,\eta, \varsigma ) - m_2 \, \rho^2 \, (\xi_1 - \eta_d)\\
\dot \varsigma &=& \varphi(\varsigma, \xi_2, \eta )\\
\theta &=& \gamma(\varsigma)
\end{array}
\end{equation}
where  $m_1,\,m_2>0$ are arbitrary, $\rho>0$ is a design parameter, $\dot\psi:\R^\de\x\R^{d\de}\x\cZ_\varsigma\to\R^\de$ is a function fixed below, and  $(\vhi,\gamma)$ is chosen to satisfy the following requirement.        
\begin{requirement}[Identifier Requirement]\label{def:identifierRequirement}
The pair $(\vhi,\gamma)$ is said to satisfy the identifier requirement relative to a class $\cH\sr$ and a cost functional \eqref{def:cJ}, if  $\vhi$ is locally Lipschitz, $\gamma$ is Lipschitz and differentiable with locally Lipschitz derivative, and there exist $\beta_\varsigma \in\cK\cL$,  a compact set $S\sr\subset\cZ_\varsigma$, $\alpha_\varsigma>0$ and, for each  $(\dot\eta^\star_d,\eta\sr)\in\cH\sr$,   a unique $\varsigma^\star:\Rplus\to S\sr$, such that:
\begin{itemize}
	\item for every locally integrable $\delta': \R\to\R^{\de}$ and $\delta'': \R\to\R^{d\de}$, all the maximal solutions to the  system  $\dot \varsigma = \varphi(\varsigma, \dot \eta_d^\star + \delta',  \eta^\star + \delta'' )$ are complete and satisfy
	\begin{equation*}
	|\varsigma(t)-\varsigma^\star(t)|\le\beta_\varsigma(|\varsigma(0)-\varsigma^\star(0)|,t)+\alpha_\varsigma \, |(\delta', \delta'')|_{[0,t)}
	\end{equation*}
	for all $t\in\Rplus$; 
	
	\item the signal $\theta^\star(t):=\gamma(\varsigma^\star(t))$ satisfies $\theta^\star(t)\in\vthb_{(\dot\eta\sr_d,\eta^\star)}^\circ(t)$ for all $t\in\Rplus$.  
\end{itemize} 
\end{requirement}
With $H\sr$ and $S\sr$ the compact sets introduced, respectively, in A\ref{ass:psiLipschitz} and in the identifier requirement, and with $\jacobian\psi$ and $\jacobian\gamma$ denoting the Jacobian of $\psi$ and $\gamma$ respectively, we define  $\dot\psi$ as any bounded function satisfying 
\begin{equation}\label{dotpsi}
\dot \psi(\xi_2,\eta, \varsigma)=\jacobian\psi(\eta,\theta)  \col\big(\Phi(\eta, \gamma(\varsigma)),\,    \jacobian\gamma (\vsig) \varphi(\varsigma, \xi_2, \eta )\big)
\end{equation}
for all $(\xi_2,\eta,\vsig)\in H\sr\x S\sr$. With this construction, since under A\ref{ass:psiLipschitz} and the identifier requirement, $\psi$, $\jacobian\psi$, $\vhi$, $\gamma$ and $\jacobian\gamma$ are locally Lipschitz, and $\dot \psi$ is bounded,  there exists $l_\psi>0$ such that 
\begin{equation}
\label{cond:dotpsi}
|\dot\psi(\xi_2,\eta,\vsig)-\dot\psi(\dot\eta\sr_d,\eta\sr,\vsig\sr)|\le l_\psi|(\xi_2-\dot\eta\sr_d,\eta-\eta\sr,\vsig-\vsig\sr)|
\end{equation}
for all $(\xi_2,\eta,\vsig)\in\R^{\de}\x\R^{d\de}\x\cZ_\varsigma$ and  $(\dot\eta\sr_d,\eta\sr,\vsig\sr)\in H\sr\x S\sr$.
 
The identifier \eqref{s:xi_vsig} is thus composed of the two subsystems $\xi$  and  $\vsig$. The dynamics and output maps $(\vhi,\gamma)$ of $\vsig$ are designed to fulfil the identifier requirement. When driven by the ``ideal'' input pair $(\dot \eta^\star_d, \eta^\star)$, the subsystem $\vsig$ is supposed to have an attractive steady-state solution $\varsigma^\star$ along which its  output $\theta\sr$ leads to the best model in the model set $\cal M$ according to (\ref{def:cJ}). In addition, a robustness property, given in terms of input-to-state stability with respect to the additive inputs $(\delta', \delta'')$, is required. This additional property is needed  since $(\dot\eta\sr_d,\eta\sr)$ is not available for feedback and  in \eqref{s:xi_vsig} the system $\vsig$ is instead driven by the input   $(\xi_2,\eta)$, the latter playing the role of a ``proxy'' for $(\dot\eta\sr_d,\eta\sr)$. While it is clear that $\eta$ carries some information on $\eta\sr$, the fact that $\xi_2$ acts as a proxy of $\dot\eta\sr_d$ follows by the definition of $\xi$, which is indeed designed as a \emph{derivative observer} of the derivative $\dot\eta_d$ of $\eta_d$, providing the missing information on $\dot\eta_d\sr$.

We stress that the ability to construct an identifier satisfying the requirement as indicated above  hides the need of qualitative and quantitative knowledge on the ideal steady-state signals $\dot\eta_d\sr$, $\eta\sr$ and $\vsig\sr$, as evident for instance in the definition of $S\sr$ and $H\sr$. We remark, however, that this information concerns high-level properties of the class $\cH\sr$, such as a uniform bound on its elements, and not the precise knowledge of the actual $(\dot\eta\sr_d,\eta\sr)$.
%	
%	
%	 Instead of driving the latter  with $(\dot \eta_d^\star,\eta^\star)$, which are not measurable,     we rather use $(\dot \eta_d, \eta)$ as proxy of the steady state. 
%
In Section \ref{sec:leastsquares},  pair $(\vhi,\gamma)$ fulfilling the identifier requirement  when the model $\psi(\cdot,\theta)$ is linearly parametrised and (\ref{def:cJ}) is a least square functional is presented.

% The high-gain  properties of the stabiliser, in fact, will guarantee that  $(\dot \eta_d, \eta)$ can be steered arbitrarily closed to  $(\dot \eta_d^\star, \eta^\star)$ regardless the value of $\theta \in {\R^\dth}$, with the asymptotic properties of (\ref{ident2}) that are recovered by the ISS condition specified in the Identifier Requirement. 
%Furthermore, since $\dot \eta_d$ cannot be used as such without generating a ``circle" effect in the identifier, we substitute it with an estimate provided by an high-gain observer of the form
% \begin{equation}\label{ddobserver}
% \ba{rcl}
%        \dot \xi_1 &=& \xi_2- m_1 \, \rho\, (\xi_1 - \eta_d)\\
%  \dot \xi_2 &=&    \dot \psi(\eta, \varsigma, \xi_2 ) - m_2 \, \rho^2 \, (\xi_1 - \eta_d)
%  \ea
%  \end{equation}
%  in which  $m_1$ and $m_2$ are such that the polynomial $s^2+m_1 s + m_2 $ is Hurwitz, $\rho$ is a design parameter and 
% \[
% \dot \psi(\eta, \varsigma,\xi_2):=\displaystyle {\partial \psi \over \partial \eta} \Phi(\eta, \gamma(\varsigma))   + {\partial \psi \over \partial \theta} {\partial \gamma \over \partial \varsigma} \varphi(\varsigma, \xi_2, \eta )\,.
% \]
%   The  identifier (\ref{sys:identifier}) is then chosen as a system with state $z = \col(\xi_1, \xi_2, \varsigma)$ and dynamics 
% \begin{equation}\label{identifierwithdd}
%   \dot \varsigma = \varphi(\varsigma, \xi_2, \eta )\,, \qquad 
%  \theta = \gamma(\varsigma)
% \end{equation}
% with $\xi_2$ generated by (\ref{ddobserver}).
 %%%%%%%%%%%%%%%%%%%%%%%%%%%
 %%%%%%%%%%%%%%%%%%%%%%%%%%%       
       \subsection{The asymptotic stability result}
     The overall regulator reads as follows
   \begin{equation}\label{s:regulator}
   \begin{array}{lcl}
   \dot \eta &=& \Phi(\eta,\gamma(\varsigma)) + Ge\\
   \dot \varsigma &=& \varphi(\varsigma,\xi_2,\eta)\\
   \dot \xi_1 &=& \xi_2 -m_1\rho(\xi_1-\eta_d)\\
   \dot \xi_2 &=& \dot\psi (\xi_2,\eta,\varsigma) - m_2\rho^2(\xi_1-\eta_d )  \\
   u &=& \cL \big(   K_\chi  \chi + K_\zeta \zeta + K_\eta \eta_1 + K_w \nu(x^\star,w)   \big)
   \end{array}
   \end{equation}    
   
      We finally show that the design parameters $(g,\ell, \kappa, \rho)$ can be chosen  so that the closed-loop system has an asymptotic regulation error that is bounded by a function of the best attainable prediction error. The result is precisely formulated in the following theorem. 
       
  \begin{theorem}\label{MainTheorem}
  	 Suppose that A\ref{ass:regulationEq} and A\ref{ass:stabilisability} hold, and consider the regulator \eqref{s:regulator} constructed in the previous sections with $\cH\sr$ and $\psi$ satisfying A\ref{ass:psiLipschitz} and $(\vhi,\gamma)$ fulfilling the identifier requirement relative to $\cH\sr$ and a cost functional \eqref{def:cJ}. Suppose moreover that $(\dot\eta\sr_d,\eta\sr)\in\cH\sr$ for all $\kappa>1$ and $\ell>1$. Then there   exist    $c,\,\rho^\star,\, g^\star(\rho),\, \kappa^\star(g),\, \ell^\star(g,\kappa)>0$ such that, for all $\rho\ge\rho\sr$, $g \geq g^\star(\rho)$, $\kappa \geq \kappa^\star(g)$ and $\ell \geq \ell^\star(g,\kappa)$, every solution of the closed-loop system \eqref{normal_form}, \eqref{s:regulator}   satisfies
  	\begin{equation*}
  	   \limsup_{t \to \infty}   |e(t)| \leq     \dfrac{c}{g^{d}}   \limsup_{t \to \infty}   |\varepsilon^\star(t, \theta^\star(t))|,
  	\end{equation*} 	
  	with $c$ not dependent on the control parameters.
\end{theorem}
%\LM{We observe that the result holds only as long as the trajectories of the closed-loop dynamics remain in a given (arbitrary large) compact set $X \times M \times \Xi$ with the lower bound of the high-gain design parameters dependent on it. Standard high-gain arguments typically used in the semiglobal stabilisation literature, omitted for reasons of space, can be used to show that the control parameters can be chosen to ensure such a boundedness property with an arbitrary compact set of initial conditions.}
Theorem \ref{MainTheorem} is proved in the Appendix. Its claim is an \emph{approximate} regulation result, which becomes \emph{asymptotic} whenever $\vep\sr(t,\theta\sr(t))=0$. This, in turn, happens when a ``real'' model exists and belongs to the chosen model set $\cM$. As Assumption A\ref{ass:psiLipschitz} and the identifier requirement imply that $\vep\sr$ can be bounded uniformly in the control parameters,  the claim of the theorem is also a \emph{practical} regulation result, with the bound on the regulation error that can be reduced arbitrarily by increasing $g$.   
  Finally, we remark that, if a ``saturated version'' of $\psi$ is implemented in the internal model unit \eqref{NLinternalModel} in place of $\psi$ (for instance by saturating $\psi$ on $H\sr\x\gamma(S\sr)$ in the same way as it is done in \eqref{dotpsi} for $\dot\psi$), and if $(\dot\eta_d\sr,\eta\sr)$ is bounded uniformly in the control parameters (which is always true in the square case as remarked in Section \ref{sec:designIM}), then a \emph{practical} regulation result is still preserved\footnote{This can  be deduced by the proof of Theorem \ref{MainTheorem} by neglecting the identifier's dynamics and by noticing that, in \eqref{pf:teta}, $\tilde\psi(\tilde\eta,\tilde\vsig,\eta\sr,\vsig\sr)-\vep\sr= \psi(\eta,\gamma(\vsig))-\dot\eta\sr_d$  can be bounded uniformly in $\vsig$.} also in the case in which $(\dot\eta\sr_d,\eta\sr)\notin\cH\sr$, thus paralleling the ``canonical'' pre-processing results (see e.g.  \cite{Byrnes2004,Isidori2012adapt}). In this case, however, the asymptotic bound on $e(t)$ cannot be related to $\vep\sr$ any more. 	
\section{Continuous-Time Least Squares Identifiers}\label{sec:leastsquares} 
We develop here an example of a pair $(\vhi,\gamma)$ that fulfils the identifier requirement when the model $\psi(\cdot,\theta)$ is a finite {\em linear combination of known functions} of the form\footnote{For ease of exposition we present here the case in which $\de=1$, with the remark that an identifier of the same kind for $\de>1$ can be always obtained as the composition of $\de$ single-variable identifiers. } 
\begin{equation}\label{ls:phi}
\psi(\cdot,\theta) = \sum_{i=1}^{\dth} \theta_{i} \sigma_{i}(\cdot)\,, 
\end{equation}
in which $\dth\in\N$ is arbitrary and $\sigma_{i}:\R^{d}\to\R$ are known Lipschitz and bounded functions. In this case the model set $\cM$ is the family of functions of the form $\sigma(\cdot)^\top \theta$, having defined $\sigma(\cdot):=\col(\sigma_1(\cdot),\dots,\sigma_{\dth}(\cdot))$ and $\theta:=\col(\theta_1,\dots,\theta_{\dth})$.   
We associate with $\cM$ the following cost functional, obtained  by letting in \eqref{def:cJ} $c_\varepsilon(t,s,\cdot) := \lambda \exp(-\lambda(t-s))|\,\cdot\,|^2$ and $c_r(\theta):=\theta^\top \Gamma\theta$, with $\lambda>0$ and  $\Gamma\in\R^{\dth\x \dth}$ symmetric and positive semi-definite
\begin{equation}\label{ls:J}
\begin{aligned}
\cJ_{(\dot\eta\sr_d,\eta^\star)}(t)(\theta) &= \lambda\int_0^t e^{-\lambda(t-s)} \big|\varepsilon^\star(s,\theta)\big|^2ds  + \theta^\top \Gamma\theta 
\end{aligned}
\end{equation}
in which the prediction error \eqref{d:varepsilon_star} at time $s$ reads as
\begin{equation*}
\varepsilon^\star(s,\theta):=\dot \eta^\star_d(s)- \sigma(\eta^\star(s))^\top \theta\,.
\end{equation*}
The optimisation problem associated with \eqref{ls:J} is recognised to be a (weighted) least squares problem with regularisation, in which $\lambda$ and $\Gamma$ play the role of the forgetting factor and the regulariser respectively. Namely, except for the regularisation term, minimising \eqref{ls:J} means minimising a weighted squared ``norm'' of the prediction errors associated with all the past data. 

With $\SP_{\dth}$ the space of symmetric positive semi-definite matrices in $\R^{\dth\x \dth}$, we let $\cZ_\varsigma:=\SP_{\dth}\x\R^{\dth}$ and, by partitioning the state as $\varsigma=(\vsig_1,\vsig_2)$, with $\varsigma_1\in\SP_{\dth}$ and $\varsigma_2\in\R^{\dth}$, we equip $\cZ_\varsigma$ with the norm $|\vsig|:=|\vsig_1|+|\vsig_2|$. We thus construct a pair $(\vhi,\gamma)$ satisfying the identifier requirement relative to \eqref{ls:J} as follows
\begin{equation}\label{ls:sys}
\begin{aligned}
&\begin{array}{lcl}
\dot \varsigma_1 &=& -\lambda \varsigma_1 + \lambda \sigma(\eta)\sigma(\eta)^\top \\
\dot \varsigma_2 &=& -\lambda \varsigma_2 + \lambda\sigma(\eta) \xi_2 \\
\theta &=& (\varsigma_1+\Gamma)\inv \varsigma_2,
\end{array} & \vsig\in\cZ_\varsigma 
\end{aligned}
\end{equation}
The claim is formalized by the following proposition.
\begin{proposition}\label{prop:ls} 
	With $c>0$ arbitrary, let $\cH\sr$ be a class of locally integrable functions $(\dot\eta\sr_d,\eta\sr):\Rplus\to\R\x\R^{d}$ satisfying $|(\dot\eta\sr_d,\eta\sr)|_\infty\le c$. Then, if $\Gamma>0$, the pair $(\vhi,\gamma)$ constructed in \eqref{ls:sys} satisfies the identifier requirement locally\footnote{The word "locally" in the claim of the proposition refers to the fact that $\gamma$ is only proved to be locally Lipschitz, and not globally Lipschitz as requested by the identifier requirement. Nevertheless, we remark that a globally Lipschitz $\gamma$ can be simply obtained by saturating the expression in (20) on the compact set in which $\varsigma$ is supposed to range, the latter that can be inferred by the knowledge of the bound $c$ on $(\dot\eta_d^\star,\eta^\star)$, as specified in the proof. Details are omitted for reason of space.} relative to $\cH\sr$ and the least-squares functional (\ref{ls:phi}) with $\beta_\varsigma(s,t)=s\exp(-\lambda t)$.  
\end{proposition}
\begin{proof} 
As $\sigma$ is Lipschitz and bounded, then $\vhi(\varsigma,\xi_2,\eta):=(-\lambda\vsig_1+\lambda\sigma(\eta)\sigma(\eta)^\top,\, -\lambda\vsig_2+\lambda\sigma(\eta)\xi_2)$ is locally Lipschitz.	  Pick an eigenvalue $\epsilon(t)$ of $\vsig_1(t)+\Gamma$, and let $v(t)\ne 0$ be a corresponding eigenvector. Then $v(t)^\top(\vsig_1(t)+\Gamma)v(t)=\epsilon(t) |v(t)|^2$, and since $\Gamma>0$ and $\vsig_1(t)\in\SP_{\dth}$, this implies $\epsilon(t)\ge p$, with $p>0$ the smallest eigenvalue of $\Gamma$. Thus $\vsig_1+\Gamma$ is invertible and the singular values of $(\vsig_1+\Gamma)\inv$ are bounded by $p\inv$, which implies that $\gamma(\vsig):=(\vsig_1+\Gamma)\inv \vsig_2$ is locally Lipschitz, smooth in $\varsigma$ and,  as a consequence,  its derivative is locally Lipschitz.

Pick now $\xi_2=\dot\eta\sr_d+\delta'$ and $\eta=\eta\sr+\delta''$, with $(\dot\eta\sr_d,\eta\sr)\in\cH\sr$ and $(\delta',\delta'')$ locally integrable. Forward completeness follows by noticing that \eqref{ls:sys} is a stable linear system driven by the locally integrable input $(\sigma(\eta)\sigma(\eta)^\top,\sigma(\eta)\xi_2)$ and that, as $\sigma(\eta)\sigma(\eta)^\top\in\SP_{\dth}$, then $\SP_{\dth}$ is forward invariant for $\varsigma_1$. With $\Sigma(\eta^\star,\dst''):= \sigma(\eta^\star+\dst'')\sigma(\eta^\star+\dst'')^\top$ and $\pi(\eta^\star, \dot \eta_d^\star,  \dst', \dst'') := \sigma(\eta^\star+\dst'')(\dot \eta_d^\star+\dst')$, 
 define 
\begin{equation*}
\begin{aligned}
\varsigma_1^\star(t) &:= \lambda \int_0^t e^{-\lambda(t-s)} \Sigma(\eta^\star(s),0) ds\\
\varsigma_2^\star(t) &:= \lambda\int_0^t e^{-\lambda(t-s)}\pi(\eta^\star(s),\dot \eta_d^\star(s),0,0)ds,
\end{aligned}
\end{equation*}
and let $\varsigma^\star=(\varsigma_1^\star,\varsigma_2^\star)$.
If $|(\dot\eta\sr_d,\eta\sr)|\le c$ for some $c>0$, then clearly there exists $c'>0$ such that $\varsigma\sr(t)\in S\sr:=\{ \vsig\in\cZ_\vsig\st |\vsig|\le c' \}$. Furthermore, since  $\sigma$ is  Lipschitz and bounded, there exists $l_\sigma>0$ (possibly depending on $c$) such that $|\Sigma(\eta^\star,\dst'')-\Sigma(\eta^\star,0)|  \le  l_\sigma |\dst''|$ and $|\pi(\eta^\star,\dot \eta_d^\star, \dst', \dst'')-\pi(\eta^\star,\dot \eta^\star_d, 0,0 )|  \le   l_\sigma |(\dst', \dst'')|$   for all $(\dst', \dst '') \in\R \times \R^{d}$. 
Hence, by integration of (\ref{ls:sys}), and using $\varsigma_1^\star(0)=0$, we obtain
\begin{equation*}
|\varsigma_1(t)-\varsigma_1^\star(t)|  \le e^{-\lambda t}|\varsigma_1(0)-\varsigma_1^\star(0)| +  { l_\sigma}  |(\delta',\delta'')|_{[0,t)},
\end{equation*}
and a similar bound holds for $|\varsigma_2(t)-\varsigma_2^\star(t)|$, thus implying the first item of the identifier requirement  with $\beta_\varsigma(s,t)=s \exp(-\lambda t)$ and with $\alpha_\varsigma=2l_\sigma$. 
For fixed $t\in\Rplus$, differentiating \eqref{ls:J} with respect to $\theta$ yields
$\jacobian_\theta  \cJ_{(\dot\eta\sr_d,\eta^\star)}(t)(\theta) =2( (\varsigma^\star_1(t)+\Gamma)\theta - \varsigma^\star_2(t))$. 
Since, $\vartheta^\circ_{(\dot\eta\sr_d,\eta^\star)}(t):=\{ \theta\in\R^{n_\theta}\st\jacobian_\theta \cJ_{(\dot\eta\sr_d,\eta^\star)}(t)(\theta)=0\}$, then $\theta^\star (t) = (\varsigma_1^\star(t)+\Gamma)\inv \varsigma_2^\star(t)\in \vartheta^\circ_{(\dot\eta\sr_d,\eta^\star)}(t)$, which is the second item of the requirement, thus concluding the proof. %\hfill\IEEEQEDhere 
\end{proof} 

We observe that the  regularisation matrix $\Gamma>0$ plays a fundamental role in Proposition \ref{prop:ls}, as it ensures that $\vsig_1+\Gamma$ is uniformly nonsingular. However, its presence frustrates the possibility of having asymptotic regulation also when the ``right'' internal model  belongs to the model set \eqref{ls:phi}. As evident in \eqref{ls:J}, indeed, having $\Gamma\ge 0$ means that, even if $\theta$ annihilates the prediction error $\varepsilon\sr$, and thus the integral term of \eqref{ls:J}, it also produces a positive addend $\theta^\top\Gamma\theta$, thus possibly making such $\theta$   a non-stationary point of $\cJ_{(\dot\eta\sr_d,\eta^\star)}(t)$. In this case, $\theta$ approaches a neighbourhood of $\theta\sr(t)$ of a size that depends on the maximum eigenvalue of $\Gamma$ that, however, can be taken as small as desired.
Nevertheless, $\Gamma$ can be  chosen positive semi-definite (and possibly zero). In this case, \eqref{sys:identifier} can still be used by substituting the inverse operator with a pseudo-inverse (indeed $\sigma_1+\Gamma$ needs not be invertible in this case), and the claim of Proposition \ref{prop:ls} applies only if the minimum non-zero singular value of $\varsigma_1+\Gamma$ is bounded away from zero uniformly in $t$, which can be seen as a \emph{persistence of excitation condition}. We also remark that, in this case, the Lipschitz constant of $\gamma$ and its derivative  becomes dependent on how large is the minimum non-zero singular value of $\varsigma_1+\Gamma$, thus making the result of Theorem \ref{MainTheorem} obtained for a certain value of the gains $\rho$, $g$, $\kappa$ and $\ell$, applicable only to the solutions carrying sufficient excitation.

\section{Example: Control of the VTOL}\label{sec:VTOL} 
Consider the lateral $(\xx_1,\xx_2)$ and angular $(\xx_3,\xx_4)$ dynamics of a VTOL aircraft described by \cite{Isidori2003Book}
\begin{equation*}%\label{x:plant}
\begin{aligned}
&\begin{array}{lcl}
\dot \xx_1 &=&\xx_2\\
\dot \xx_2 &=& {\rm d}(w) -  \gravity\tan \xx_3  + v \end{array}& &\begin{array}{lcl} 
\dot \xx_3&=& \xx_4\\
\dot \xx_4 &=& B u
\end{array}
\end{aligned}
\end{equation*}
with  $\gravity>0$ the gravitational constant and $B=2LJ\inv>0$, with $L>0$ the length of the wings and $J$ the moment of inertia (typically uncertain).  The input $u$ is the force on the wingtips, $v$ is a vanishing input taking into account the (controlled) vertical dynamics (not considered here) and ${\rm d}(w):=M\inv {\rm d}_0(w)$, with ${\rm d}_0(w)$  the lateral wind force disturbance, and $M>0$ the VTOL mass. The control goal is to eliminate the wind action from the lateral position dynamics, i.e. the regulation error is defined as $e(t)=\xx_1(t)$. We also suppose to have available for feedback the entire state, namely $y=\xx$. Let $w$ be generated by an exosystem of the form \eqref{fmk:sys:exo} and change variables as $\xx \mapsto x:=(\chi,\zeta)$, with $\chi:=(\xx_1,\,\xx_2,\,-\gravity\tan \xx_3+{\rm d}(w))$ and $\zeta:= L_s{\rm d}(w) - \gravity \xx_4/ (\cos \xx_3)^2$. In the new coordinates the following equations hold 
\begin{equation*}
\begin{aligned}
\begin{array}{lcl}
\dot \chi_1 &=&\chi_2\\
\dot\chi_2 &=& \chi_3  
\end{array}&&\begin{array}{lcl}
\dot\chi_3 &=& \zeta\\
\dot\zeta &=& q(w,x) + \Omega(w,x) u ,
\end{array}
\end{aligned}
\end{equation*}
in which\footnote{Recall that $\cos(\tan\inv(s))=1/\sqrt{s^2+1}$.} $\Omega(w,x):=-\gravity B/(\cos(\tan\inv({\rm d}(w)-\chi_3)/\gravity)))^2$ and $q(w,x)$  properly defined. This system is in the form (\ref{normal_form}), with A\ref{ass:regulationEq} trivially fulfilled ($x_0$ being absent) by $x\sr=0$ and $u\sr=\gravity B\inv (L_s^2 {\rm d}(w)-2 {\rm d(w)}^2 L_s {\rm d(w)} )/({\rm d(w)}^2+\gravity^2)$,  and A\ref{ass:stabilisability} fulfilled on each compact set with ${\cal L}$ a negative number\footnote{In this respect, we observe that the ideal steady-state value of the measurements $(\xx_3,\xx_4)$ is given by $(\xx_3\sr,\xx_4\sr):=(\tan\inv({\rm d}(w)/\gravity),\gravity L_s{\rm d}(w)/({\rm d}(w)^2+1))$, and thus $y$ is not in general vanishing at the steady state.}.
With $(c_1, c_2,c_3)$ the coefficients of a Hurwitz polynomial  and  $\kappa,\ell>0$ design parameters,  we fix the control law as 
\begin{align*}
u = -{\cal L} \big( &c_1 \ell \kappa ^3 (\xx_1+\eta_1) +c_2 \ell \kappa ^2 \xx_2  +c_3 \ell \kappa  (-\gravity\tan\xx_3) \\&\; + \ell (-\gravity \xx_4/\cos^2\xx_3) \big),
\end{align*}
with $\eta_1$ the first state of the internal model fixed later. In the new  coordinates $(\chi, \zeta)$, this control law is of the form (\ref{stabiliser}), with ${\cal K}_w= \ell(c_3\kappa\;\; 1)$ and $\nu(x^\star,w)=\col({\rm d}(w), \, L_s{\rm d}(w))$.

Regarding the design of the internal model unit, we observe that, by following Section \ref{sec:designIM}, $\Upsilon_{(\ell,\kappa)}(w) = Q(\ell, \kappa)  D(w)$,
in which  $Q(\ell, \kappa):=({c_3}/{(c_1\kappa^2)} \; {1}/{(c_1\kappa^3)} \;  {1}/{(c_1\ell {\cal L}\kappa^3)} )$ and $D(w) =\col({\rm d}(w), \, L_s{\rm d}(w), \, -\Omega(w,0)\inv q(w,0))$. Thus, $\eta_i^\star = Q(\ell, \kappa)  L_s^{i-1}D(w)$, $i=1,\ldots, d$, and $\dot \eta_d^\star = Q(\ell, \kappa)  L_s^{d}D(w)$. The form of $Q$ and the fact $\kappa$ and $\ell$ have large values   show that the dominant elements in $\eta^\star_i$ and $\dot \eta^\star_d$ are $L_s^{i-1} {\rm d}(w)$ and  $L_s^{d} {\rm d}(w)$, regardless the value of the dimension $d$ of $\eta$. Now, suppose that ${\rm d}(w)$ consists of a single harmonic at an unknown frequency. The design of $(d,\psi)$ and the identifier to reject ${\rm d}(w)$ is then carried out by considering a single oscillator as the model set, obtained with $d=2$, $\dth=2$, and $\psi(\eta,\theta):=\theta^\top \eta$. The adaptation phase, in turn, can be set up by using the least-squares identifier presented in Section \ref{sec:leastsquares} with $n_\theta=2$ and $\sigma$ any bounded function satisfying $\sigma(\eta)=\eta$ in the region where $\col({\rm d}(w),L_s{\rm d}(w))\cdot c_3/c_1\kappa^2$ is supposed to range.  
%
%By way of example, if we   the wind shows a dominant (yet \emph{unknown}) harmonic, 

% \subsection{{\color{blue} Mettere eventualmente anche uno sketch dell'esempio multivariabile se ci sta} }
%
      
\section{{Conclusions}}\label{sec:conclusions} 
The paper presented a post-processing design procedure for a class of multivariable nonlinear systems stabilisable by high-gain feedback   hinging on a  ``non-equilibrium'' framework. The internal model is adaptive with the adaptation mechanisms cast as an identification problem and with the asymptotic regulation error that is directly related to the identification error.  The framework does not rely on an exact knowledge of the steady state friend, nor on an exact parametrisation of it. Rather, it assumes the knowledge of some qualitative\textbackslash{}quantitative information about the class of steady state signals used to choose the model set of the underlying identification problem.
 The  paper fits in the research direction of   \cite{DAstolfi2015,Bin2018c}  in which approximate, rather than asymptotic, regulation is envisioned as the right perspective in presence of general uncertainties. 
%Future works regard the extension of the technique to a larger class of nonlinear systems, not necessarily stabilisable by high-gain feedback and the use of different identification tools. 

\appendices
%%%%%%%%%%%%%%%%%%%%%%%%%%%%%%%%%%%%%%%%%%
%%%%%%%%%%%%%%%%%%%%%%%%%%%%%%%%%%%%%%%%%%

%%%%%%%%%%%%%%%%%%%%%%%
%%%%%%%%%%%%%%%%%%%%%%%
\section{Proof of Theorem \ref{MainTheorem}}\label{sec:ProofPropNf}
With $\kappa>1$ and $\ell>1$, pick a solution $(x,\chi,\zeta,\eta,\varsigma,\xi)$ to the closed-loop system \eqref{normal_form}, \eqref{s:regulator} and let $(x^\star,u^\star,\eta^\star,\dot\eta_d^\star)$ be given by A\ref{ass:regulationEq} and   Section \ref{sec:designIM}. Assume that $(\dot\eta\sr_d,\eta\sr)\in\cH\sr$, and let $(\varsigma^\star,\theta^\star)$ be produced by the identifier requirement. Consider the following change of variables
\begin{equation*}
\begin{array}{lcllcl}
\eta  &\mapsto &  \tilde\eta   :=  \eta -\eta ^\star   &
\varsigma &\mapsto&  \tilde\varsigma  :=  \varsigma-\varsigma^\star\\
\chi  & \mapsto  & \tilde\chi   :=  \chi  + C^\top \tilde\eta_1 & \zeta&\mapsto &  \tilde\zeta  :=  \zeta- K(\kappa)\tilde\chi  \\
\xi&\mapsto & \tilde\xi := \xi-\begin{pmatrix}
 \eta_d\sr\\\psi(\eta\sr,\theta\sr) 
\end{pmatrix}&
e &\mapsto &  \tilde e   :=  e+\tilde\eta_1, 
\end{array}
\end{equation*}
where we recall that $K(\kappa)$ is defined in \eqref{stabiliser} and $\tilde\eta_1\in\R^\de$ represents the first $\de$ components of $\tilde\eta$. By definition of $\eta^\star$, $\dot\eta\sr_i=\eta\sr_{i+1}$, and in the new coordinates we obtain
\begin{equation}\label{pf:teta}
\begin{array}{lcl}
\dot{\tilde\eta}_i &=& \tilde\eta_{i+1} - h_ig^i \tilde\eta_1 + g^ih_i \tilde e, \qquad i=1,\dots,d-1\\
\dot{\tilde\eta}_d &=&-h_dg^d\tilde\eta_1 + \tilde\psi(\tilde\eta,\tilde\varsigma,\eta^\star,\varsigma\sr) +h_dg^d\tilde e - \varepsilon\sr.
\end{array}
\end{equation}
with $\vep\sr=\vep\sr(t,\theta\sr)$ given by \eqref{d:varepsilon_star} and with $\tilde\psi(\tilde\eta,\tilde\varsigma,\eta^\star,\varsigma\sr):=  \psi(\tilde\eta+\eta^\star,\gamma(\tilde\varsigma+\varsigma\sr)) - \psi(\eta\sr,\gamma(\varsigma\sr))$ that, since A\ref{ass:psiLipschitz} and the identifier requirement imply that $\psi$ and $\gamma$ are Lipschitz, fulfils 
%\begin{equation*}%\label{pf:tpsi}
$|\tilde\psi(\tilde\eta,\tilde\varsigma,\eta^\star,\varsigma\sr)|\le  c_{\psi,\gamma} |(\tilde\eta,\tilde\vsig)|$
%\end{equation*} %
for some $c_{\psi,\gamma}>0$ independent on the control parameters.
Thus, standard high-gain arguments (see e.g. \cite{KhalilNL}) show that there exist $a_0,a_1,a_2,a_3>0$ and $g^\star_0>0$ such that for all $g\ge g^\star_0$ the following bound holds
\begin{equation}\label{pf:bound_eta} 
\begin{aligned}
%|\tilde\eta(t)| &\le  \dfrac{a_0}{g^{1-d}}   |\tilde \eta(0)|\e^{-a_1gt} +  \dfrac{a_2}{g}    |(\tilde\varsigma, \vep\sr)|_{[0,t)}  + a_3g^{d-1} |\tilde e|_{[0,t)}  \\
 |\tilde\eta_i(t)| &\le   a_0 g^{i-1} |\tilde \eta(0)|\e^{-a_1gt} +   a_2 g^{i-d-1}     |(\tilde\varsigma,\vep\sr)|_{[0,t)}\\&\qquad  +a_3 g^{i-1}  |\tilde e|_{[0,t)} 
\end{aligned}
\end{equation}
for all $t\in\Rplus$ and each $i=1,\dots,d$. Moreover, $\tilde\xi$ satisfies
\begin{equation*}
\begin{array}{lcl}
\dot{\tilde\xi}_1 &=& \tilde\xi_2 -m_1\rho\tilde\xi_1 + m_1\rho\tilde\eta_d-\vep\sr \\
\dot{\tilde\xi}_2&=& -m_2\rho^2\tilde\xi_1 + \tilde{\mu}(\tilde\eta,\tilde\varsigma,\tilde\xi_2,\eta\sr,\varsigma\sr) + m_2\rho^2\tilde\eta_d ,
\end{array}
\end{equation*}
in which, since by A\ref{ass:psiLipschitz}  $(\dot\eta\sr_d,\eta\sr)\in\cH\sr$ implies $(\dot\eta\sr_d(t),\eta\sr(t))\in H\sr$, and by the identifier requirement we have $\vsig\sr(t)\in  S\sr$,   in view of \eqref{dotpsi} $\tilde \mu$ reads as $\tilde\mu(\tilde\eta,\tilde\varsigma,\tilde\xi_2,\eta\sr,\varsigma\sr)  := \dot\psi(\tilde\xi_2+\psi(\eta\sr,\theta\sr),\tilde\eta+\eta\sr,\tilde\varsigma+\varsigma\sr) -\dot\psi(\dot\eta_d\sr,\eta\sr,\varsigma\sr)$.
 Moreover, in view of \eqref{cond:dotpsi}, there exists $l_\psi>0$ such that $|\tilde\mu(\tilde\eta,\tilde\varsigma,\tilde\xi_2,\eta\sr,\varsigma\sr)|\le l_\psi\big( |(\tilde\eta,\tilde\varsigma,\tilde\xi_2)| + |\vep\sr|  \big)$. 
 Hence,  customary high-gain arguments show that there exist $\rho \sr_0>1$ and $a_4,a_5,a_6>0$ such that, for all $\rho\ge\rho \sr_0$, the following  holds
\begin{equation}\label{pf:bound_xi}
\begin{aligned}
|\tilde\xi(t)|&\le  a_4\rho|\tilde\xi(0)|\e^{-a_5\rho t}  + a_6\left( \rho |\tilde\eta |_{[0,t)} +  \rho\inv |\tilde\varsigma|_{[0,t)}  +  |\vep\sr|_{[0,t)} \right).
\end{aligned}
\end{equation} 

We can write $\xi_2=\dot{\eta}_d\sr + \delta'$ and $\eta=\eta\sr+\delta''$, with $\delta':=\tilde\xi_2 -\vep\sr $ and $\delta'':=\tilde\eta$,
%\begin{align*}
%\xi_2 & = \dot\eta_d\sr + \tilde\xi_2 -\vep\sr & \eta &=\eta\sr+\tilde\eta, 
%\end{align*}
so that the identifier requirement yields
\begin{equation}\label{pf:bound_varsig}
|\tilde\varsigma(t)|\le \beta_\varsigma(|\tilde\varsigma(0)|,t) + \alpha_\varsigma|(\tilde\eta,\tilde\xi,\varepsilon\sr)|_{[0,t)}.
\end{equation}
In view of standard small-gain arguments (see e.g. \cite{Jiang1994}), the bounds \eqref{pf:bound_eta}, \eqref{pf:bound_xi}, \eqref{pf:bound_varsig} yield the existence of $\beta_1\in\cK\cL$, $a_7>0$, $\rho\sr\ge \rho\sr_0$ and  $g^\star(\rho)\ge g^\star_0$ such that, for all $\rho>\rho\sr$ and  $g\ge g^\star(\rho)$, we have
\begin{equation}\label{pf:bound_id}
\begin{aligned}
 &|(\tilde\eta(t),\tilde\varsigma(t),\tilde\xi(t))|\le  \beta_1(|(\tilde\eta(0),\tilde\varsigma(0),\tilde\xi(0))|,t)\\&\hspace{7em} + a_7\big(g^{d-1} |\tilde e|_{[0,t)} + |\vep\sr|_{[0,t)}\big)
 \\
&|\tilde\eta_i(t)|\le  \beta_1(|(\tilde\eta(0),\tilde\varsigma(0),\tilde\xi(0))|,t) \\&\hspace{7em}  + a_7  \big( g^{i-1}|\tilde e|_{[0,t)} + g^{i-1-d}|\vep\sr|_{[0,t)}\big).%\\
%&|\tilde\eta_1(t)|\le \beta_1(|(\tilde\eta(0),\tilde\varsigma(0),\tilde\xi(0))|,t) + a_7  \big( |\tilde e|_{[0,t)} + g^{-d}|\vep\sr|_{[0,t)}\big) 
\end{aligned}
\end{equation} 
By noticing that $\tilde e=C\tilde\chi$, differentiating $\tilde\chi$  yields
\begin{equation*}%\label{m:s:tilde_chi}
\dot{\tilde \chi} = (F+HK(\kappa)+gh_1 C^\top C)\tilde \chi + H \tilde \zeta +  C^\top (\tilde\eta_2-gh_1\tilde\eta_1),
\end{equation*}
so that, in view of \eqref{d:K}, quite standard high-gain arguments (see e.g. \cite{IsidoriNCS2}) show that  there exists $\kappa^\star_0(g)>1$ such that, for all $\kappa>\kappa^\star_0(g)$ the following  hold 
\begin{equation}\label{pf:bound_chi}
\begin{aligned}
&|\tilde\chi(t)|\le a_9(\kappa) |\tilde\chi(0)|\e^{-a_{10}\kappa t} + \dfrac{a_{11}}{\kappa} |\tilde\zeta|_{[0,t)} \\&\qquad\qquad\qquad\qquad\quad+ a_{12}(\kappa)\big( g|\tilde\eta_1|_{[0,t)}+|\tilde\eta_2|_{[0,t)}\big)    \\%\label{m:ineq:lemma_claim_iss}\\
&|\tilde e(t)| \le a_{9}(\kappa) |\tilde\chi(0)|\e^{-a_{10}\kappa t} + \dfrac{a_{11}}{\kappa} |\tilde\zeta|_{[0,t)} \\&\qquad\qquad\qquad\qquad\quad+ \dfrac{a_{13}}{\kappa} \big( g|\tilde\eta_1|_{[0,t)}+|\tilde\eta_2|_{[0,t)}\big)   %\label{m:ineq:lemma_claim_oiss}
\end{aligned}
\end{equation}
for some $a_9(\kappa),a_{10},a_{11},a_{12}(\kappa),a_{13}>0$.
Furthermore,   in the new coordinates, the control law \eqref{stabiliser} becomes $u = -\ell \cL \tilde\zeta - \cL \big( \Omega(w,x\sr)\cL\big)\inv q(w,x\sr)$, and differentiating $\tilde\zeta$ yields
% \begin{equation}
%\dot{\tilde \zeta} =   \delta(\tilde\eta,\tilde\chi,\tilde\zeta) + q(w,x)+ \Omega(w,x) u
%\end{equation}
%
%Developing further \eqref{pf:tilde_zeta} yields
\begin{equation}\label{pf:tilde_zeta}
\dot{\tilde{\zeta}} = \delta(\tilde\eta,\tilde\chi,\tilde\zeta) + \tilde \phi(w,x,x\sr) - \ell \Omega(w,x)\cL\tilde \zeta 
\end{equation}
with $\delta(\tilde\eta,\tilde\chi,\tilde\zeta)  := -K(\kappa) ((F+HK(\kappa)+gh_1C^\top C)\tilde\chi + H\tilde\zeta + C^\top (\tilde\eta_2-h_1g\tilde\eta_1) )$ that satisfies 
$
|\delta(\tilde\eta,\tilde\chi,\tilde\zeta)|\le \ a_{14}(\kappa,g) |(\tilde\eta,\tilde\chi,\tilde\zeta)| 
$,
for some $a_{14}(g,\kappa)>0$, and
with $\tilde \phi (w,x,x\sr):=\Omega(w,x)\cL( ( \Omega(w,x)\cL)\inv q(w,x) - (\Omega(w,x\sr)\cL)\inv q(w,x\sr)  )$ that,  in view of A\ref{ass:stabilisability} and since  $|\chi|\le|\tilde\chi|+|\tilde\eta_1|$ and $|\zeta|\le|\tilde\zeta|+|K(\kappa)\tilde\chi|$,  satisfies
$
|\tilde \phi(w,x,x\sr)| \le a_{15}(\kappa)|(\tilde x_0,\tilde\chi,\tilde\zeta,\tilde\eta_1)|,
$
for some $a_{15}(\kappa)>0$ and with $\tilde x_0:=x_0-x_0\sr$. Hence,  usual high-gain arguments show that, under A\ref{ass:stabilisability}, there exists an $\ell^\star_0(\kappa,g)>0$ such that, for all $\ell>\ell^\star_0(\kappa,g)$ the following bound holds
\begin{equation}\label{pf:bound_zeta}
\begin{aligned}
|\tilde\zeta(t)|\le &a_{16}|\tilde\zeta(0)|\e^{-a_{17}\ell t} + \dfrac{a_{18}(\kappa)}{\ell} |\tilde\chi|_{[0,t)}\\&+ \dfrac{a_{19}}{\ell}  |\tilde x_0|_{[0,t)} +
\dfrac{a_{20}(\kappa)}{\ell}\big(g|\tilde\eta_1|_{[0,t)} +|\tilde\eta_2|_{[0,t)}\big) 
\end{aligned}
\end{equation}
for some $a_{16},a_{17},a_{18}(\kappa),a_{19},a_{20}(\kappa)>0$.
Furthermore, by noticing that $|\chi|\le|\tilde\chi|+|\tilde\eta_1|$ and $|\zeta|\le|\tilde\zeta|+|K(\kappa)\tilde\chi|$, A\ref{ass:regulationEq} yields the existence of $b_2,b_3(\kappa)>0$ such that
\begin{equation}\label{pf:bound_x0}
|\tilde x_0(t)|\le \beta_0(|\tilde x_0(0)|,t) + b_2|(\tilde\eta_1,\tilde\zeta)|_{[0,t)} + b_3(\kappa)|\tilde\chi|_{[0,t)}.
\end{equation} 
Therefore, in view of \eqref{pf:bound_id}, \eqref{pf:bound_chi}, \eqref{pf:bound_zeta} and \eqref{pf:bound_x0}, repeating the small-gain arguments of \cite{Jiang1994} yields the existence of a $\kappa\sr(g)\ge \kappa_0\sr(g)$ and of an $\ell\sr(\kappa,g)\ge\ell_0\sr(\kappa,g)$ such that, for each $\rho>\rho\sr$, $g\ge g\sr(\rho)$, $\kappa\ge \kappa\sr(g)$, and $\ell\ge \ell\sr(g,\kappa)$, it holds that
\begin{align*}
&|(\tilde x(t),\tilde \eta(t),\tilde\varsigma(t),\tilde\xi(t))|\le \beta(|(\tilde x(0),\tilde \eta(0),\tilde\varsigma(0),\tilde\xi(0))|,t) \\&\qquad\qquad\qquad\qquad\qquad + p_1(\kappa,g) |\vep\sr|_{[0,t)}
\end{align*}
and
\begin{equation*}
\begin{aligned}
& \limsup_{t \to \infty}|\tilde e(t)|\le p_2 \kappa\inv g^{1-d} \limsup_{t \to \infty}|\vep\sr(t)| \\
& \limsup_{t \to \infty}|\tilde \eta_1(t)| \le p_3  g^{ -d} \limsup_{t \to \infty}|\vep\sr(t)| 
\end{aligned}
\end{equation*}
for some $\beta\in\cK\cL$ and $p_1(\kappa,g),p_2,p_3>0$, and the result follows with $c=p_2+p_3$ by noticing that $|e|\le |\tilde e| + |\tilde\eta_1|$ and that, since $\kappa\sr(g)$ can be taken to be larger than $g$, then $\kappa\inv\le g\inv$.

%%%%%%%%%%%%%%%%%%%%%%%%%%%%%%%%%%%%%%%%%%%%%%%%%

\bibliographystyle{IEEEtran}
\bibliography{IEEEabrv,biblio}

\begin{thebibliography}{10}
\providecommand{\url}[1]{#1}
\csname url@samestyle\endcsname
\providecommand{\newblock}{\relax}
\providecommand{\bibinfo}[2]{#2}
\providecommand{\BIBentrySTDinterwordspacing}{\spaceskip=0pt\relax}
\providecommand{\BIBentryALTinterwordstretchfactor}{4}
\providecommand{\BIBentryALTinterwordspacing}{\spaceskip=\fontdimen2\font plus
\BIBentryALTinterwordstretchfactor\fontdimen3\font minus
  \fontdimen4\font\relax}
\providecommand{\BIBforeignlanguage}[2]{{%
\expandafter\ifx\csname l@#1\endcsname\relax
\typeout{** WARNING: IEEEtran.bst: No hyphenation pattern has been}%
\typeout{** loaded for the language `#1'. Using the pattern for}%
\typeout{** the default language instead.}%
\else
\language=\csname l@#1\endcsname
\fi
#2}}
\providecommand{\BIBdecl}{\relax}
\BIBdecl

\bibitem{Bin2018c}
M.~Bin, D.~Astolfi, L.~Marconi, and L.~Praly, ``About robustness of internal
  model-based control for linear and nonlinear systems,'' in \emph{2018 IEEE
  57th Conf. Decision and Control}, 2018.

\bibitem{Byrnes2003}
C.~I. Byrnes and A.~Isidori, ``Limit sets, zero dynamics and internal models in
  the problem of nonlinear output regulation,'' \emph{IEEE Trans. Autom.
  Control}, vol.~48, pp. 1712--1723, 2003.

\bibitem{Huang1994}
J.~Huang and C.~F. Lin, ``On a robust nonlinear servomechanism problem,''
  \emph{IEEE Trans. Autom. Control}, vol.~39, no.~7, pp. 1510--1513, 1994.

\bibitem{Byrnes2004}
C.~I. Byrnes and A.~Isidori, ``Nonlinear internal models for output
  regulation,'' \emph{IEEE Trans. Autom. Control}, vol.~49, pp. 2244--2247,
  2004.

\bibitem{DelliPriscoli2006}
F.~D. Priscoli, L.~Marconi, and A.~Isidori, ``A new approach to adaptive
  nonlinear regulation,'' \emph{{SIAM} J. {C}ontrol {O}ptim.}, vol.~45, no.~3,
  pp. 829--855, 2006.

\bibitem{Marconi2007}
L.~Marconi, L.~Praly, and A.~Isidori, ``Output stabilization via nonlinear
  {L}uenberger observers,'' \emph{SIAM J. Control Optim.}, vol.~45, pp.
  2277--2298, 2007.

\bibitem{Serrani2001}
A.~Serrani, A.~Isidori, and L.~Marconi, ``Semiglobal nonlinear output
  regulation with adaptive internal model,'' \emph{IEEE Trans. Autom. Control},
  vol.~46, no.~8, pp. 1178--1194, 2001.

\bibitem{Isidori2012}
A.~Isidori and L.~Marconi, ``Shifting the internal model from control input to
  controlled output in nonlinear output regulation,'' in \emph{2012 IEEE 51st
  Conference on Decision and Control}, 2012, pp. 4900--4905.

\bibitem{Bin2017}
M.~Bin and L.~Marconi, ``About a post-processing design of regression-like
  nonlinear internal models,'' in \emph{20th {IFAC} {W}orld {C}ongress},
  Toulouse, France, Jul. 2017.

\bibitem{Wang2016}
L.~Wang, A.~Isidori, H.~Su, and L.~Marconi, ``Nonlinear output regulation for
  invertible nonlinear {MIMO} systems,'' \emph{Int. J. Robust Nonlinear
  Control}, vol.~26, pp. 2401--2417, 2016.

\bibitem{Davison1976}
E.~J. Davison, ``The robust control of a servomechanism problem for linear
  time-invariant multivariable systems,'' \emph{IEEE Trans. Autom. Control},
  vol.~21, pp. 25--34, Feb. 1976.

\bibitem{Bin2018b}
M.~Bin and L.~Marconi, ``The chicken-egg dilemma and the robustness issue in
  nonlinear output regulation with a look towards adaptation and universal
  approximators,'' in \emph{2018 IEEE 57th Conference on Decision and Control},
  2018.

\bibitem{DAstolfi2017}
D.~Astolfi and L.~Praly, ``Integral action in output feedback for multi-input
  multi-output nonlinear systems,'' \emph{IEEE Trans. Autom. Control}, vol.~62,
  no.~4, pp. 1559--1574, 2017.

\bibitem{DAstolfi2015}
D.~Astolfi, L.~Praly, and L.~Marconi, ``Approximate regulation for nonlinear
  systems in presence of periodic disturbances,'' in \emph{Proc. 54th IEEE
  Conference on Decision and Control}, 2015, pp. 7665--7670.

\bibitem{Hirschorn1979}
M.~R. Hirschorn, ``Invertibility of multivariable nonlinear control systems,''
  \emph{IEEE Trans. Autom. Control}, vol.~24, no.~6, pp. 855--865, 1979.

\bibitem{Singh1981}
S.~N. Singh, ``A modified algorithm for invertibility of nonlinear systems,''
  \emph{IEEE Trans. Autom. Control}, vol.~26, no.~2, pp. 595--598, 1981.

\bibitem{Wang2015}
L.~Wang, A.~Isidori, and H.~Su, ``Global stabilization of a class of invertible
  {MIMO} nonlinear systems,'' \emph{IEEE Trans. Autom. Control}, vol.~60,
  no.~3, pp. 616--631, 2015.

\bibitem{Teel1995}
A.~R. Teel and L.~Praly, ``Tools for semiglobal stabilization by partial state
  and output feedback,'' \emph{SIAM J. Control Optim.}, vol.~33, pp.
  1443--1488, 1995.

\bibitem{Liberzon2002}
D.~Liberzon, A.~S. Morse, and E.~D. Sontag, ``Output-input stability and
  minimum-phase nonlinear systems,'' \emph{IEEE Trans. Autom. Control},
  vol.~47, no.~3, pp. 422--436, 2002.

\bibitem{IsidoriNCS2}
A.~Isidori, \emph{Nonlinear Control Systems {II}}.\hskip 1em plus 0.5em minus
  0.4em\relax London, Great Britain: Springer-Verlag, 1999.

\bibitem{Ljung1999}
L.~Ljung, \emph{System Identification. {T}heory for the user}.\hskip 1em plus
  0.5em minus 0.4em\relax Prentice Hall, 1999.

\bibitem{Isidori2012adapt}
A.~Isidori, L.~Marconi, and L.~Praly, ``Robust design of nonlinear internal
  models without adaptation,'' \emph{Automatica}, vol.~48, pp. 2409--2419,
  2012.

\bibitem{Isidori2003Book}
A.~Isidori, L.~Marconi, and A.~Serrani, \emph{Robust Autonomous Guidance. An
  internal Model Approach}.\hskip 1em plus 0.5em minus 0.4em\relax
  Springer-Verlag, 2003.

\bibitem{KhalilNL}
H.~K. Khalil, \emph{Nonlinear Systems. Third Edition}.\hskip 1em plus 0.5em
  minus 0.4em\relax Prentice Hall, 2002.

\bibitem{Jiang1994}
Z.-P. Jiang, A.~R. Teel, and L.~Praly, ``Small-gain theorem for {ISS} systems
  and applications,'' \emph{Math. Control Signals Systems}, vol.~7, pp.
  95--120, 1994.

\end{thebibliography}

\end{document}